\documentclass[letterpaper, 10pt, conference]{ieeeconf}
\IEEEoverridecommandlockouts
\usepackage[utf8]{inputenc} 
\usepackage{dsfont}
\usepackage{cite}
\usepackage{amsfonts}
\usepackage{amsmath}
\usepackage{amssymb}
\usepackage{hyperref}       
\usepackage{url}            
\usepackage{booktabs}       

\usepackage{nicefrac}       
\usepackage{microtype}      

\usepackage{algorithmic}
\usepackage{mathrsfs}
\usepackage{graphicx}
\usepackage{textcomp}
\usepackage{xcolor}

\newenvironment{customproof}[1]{\par\noindent\textit{#1:}\quad}{\hfill$\blacksquare$}

\newtheorem{theorem}{Theorem}

\newtheorem{lemma}{Lemma}

\newtheorem{assumption*}{Assumption}
\newtheorem{stdassumption*}{Standing Assumption}

\newtheorem{prop}{Proposition}

\newtheorem{definition*}{Definition}

\DeclareMathOperator*{\argmin}{arg\,min}

\DeclareMathOperator*{\R}{\mathbb{R}}

\usepackage{mathtools}

\usepackage{dsfont}

\date{}
\title{On the Variational Interpretation of Mirror Play in Monotone Games }

\author{%
  Yunian Pan, Tao Li, and Quanyan Zhu
  \thanks{Authors are with the Department of Electrical and Computer Engineering, New York University, NY, 11201, USA {\tt yp1170, tl2636, qz494@nyu.edu}
}}

\begin{document}
 
\maketitle

\begin{abstract}
Mirror play (MP) is a well-accepted primal-dual multi-agent learning algorithm where all agents simultaneously implement mirror descent in a distributed fashion. The advantage of MP over vanilla gradient play lies in its usage of mirror maps that better exploit the geometry of decision domains. Despite extensive literature dedicated to the asymptotic convergence of MP to equilibrium, the understanding of the finite-time behavior of MP before reaching equilibrium is still rudimentary. To facilitate the study of MP's non-equilibrium performance, this work establishes an equivalence between MP's finite-time primal-dual path (mirror path) in monotone games and the closed-loop Nash equilibrium path of a finite-horizon differential game, referred to as mirror differential game (MDG). Our construction of MDG rests on the Brezis-Ekeland variational principle, and the stage cost functional for MDG is Fenchel coupling between MP's iterates and associated gradient updates. The variational interpretation of mirror path in static games as the equilibrium path in MDG holds in deterministic and stochastic cases. Such a variational interpretation translates the non-equilibrium studies of learning dynamics into a more tractable equilibrium analysis of dynamic games, as demonstrated in a case study on the Cournot game, where MP dynamics corresponds to a linear quadratic game.


\end{abstract}

\section{Introduction}
Mirror play (MP) refers to the multi-agent learning dynamics where all agents simultaneously run mirror descent, a generalization of gradient descent \cite{mertikopoulos2019learning}. The notable feature of MP is the usage of mirror maps that bridge the primal space of iterates and the dual space of gradients, creating gradient dynamics in the dual space, which tame the domain geometry for faster convergence \cite{nemirovskij1983problem}. As a popular primal-dual algorithm, MP also appears in the literature under the name of dual averaging \cite{mertikopoulos2019learning}, the follow-the-regularized-leader \cite{NIPS2006_1cfead99}, and EXP3 (see \cite{tao22confluence} for their connections), enjoying broad applications in learning-based decentralized control of network systems \cite{tao22confluence}.

Extensive efforts have been dedicated to MP's asymptotic behavior and its convergence to equilibrium in games, as briefly reviewed in ``Related Works.'' However, a comprehensive inspection of its finite-time iterates before reaching equilibrium, called the non-equilibrium path \cite{pan2023resilience}, remains largely uncharted, despite some early endeavors \cite{lei18md-hp, pan2023resilience, pan2023stochastic}. 

To complement the prior works on MP's non-equilibrium studies resting on martingale theory and concentration inequalities \cite{lei18md-hp, pan2023resilience, pan2023stochastic}, this work offers a novel dynamic-system perspective on MP dynamics in static monotone games, whereby the MP finite-time iterates are equivalent to the closed-loop Nash equilibrium path of a finite-horizon differential game (referred to as the mirror differential game) whose stage cost is given by the Brezis-Ekeland action functional \cite{brezis1976principe,tzen2023variational} and the terminal cost by the Bregman divergence between the last iterate to the Nash equilibrium of the monotone game.  A brief mathematical summary of our contribution is below.

\subsubsection*{\bf Contribution and Key Result} Given a set of players $\mathcal{N}$ in a monotone game (see Sec.\ref{subsec:monoton}), with the player indexed by $i\in \mathcal{N}$, employing strategy $y_i\in \mathcal{Y}_i$, and its payoff given by $\psi_i(y_i, y_{-i})$, the continuous-time MP dynamics is 
\begin{equation}
\label{eq:MP}
    \dot{x_i} = -\nabla_i \psi_i (\nabla\phi_i^*(x_i), \nabla\phi_{-i}^*(x_{-i})), x_i(0)=x_{i,0}, \tag{MP}
\end{equation}
where $\phi_i: \mathcal{Y}_i\rightarrow\R 
$ is of Legndre type (see ``Preliminaries''), and $\phi_i^*$ denotes the Fenchel conjugate. $\nabla\phi_i^*: \mathcal{X}_i\rightarrow \mathcal{Y}_i $ is the mirror map taking the dual iterate $x_i$ to the primal space $\mathcal{Y}_i$. 

Consider a differential game where players seek the equilibrium control $u_i(t)$ to minimize the cumulative cost functional subject to the system dynamics and others' strategies:
\begin{align}
\label{eq:MDG}
        J_i (x(0), \{u_i\}) &= \int_{t=0}^{T} c_i (x(t), u_i(t), u_{-i}(t)) \mathrm{d}t + q_i(x(T))\nonumber\\
        \text{s.t. }  \dot{x}_i(t)&=u_i(t), i\in \mathcal{N}, \tag{MDG}
\end{align}
where $c_i$ and $q_i$ are defined using the mirror map $\nabla\phi^*_i$ and the monotone game payoff $\psi_i$ to be introduced in Sec.~\ref{subsec:MP2MDG}. Hence, we refer to the dynamic game above as the mirror differential game \eqref{eq:MDG} with respect to \eqref{eq:MP}.

\textbf{Our contributions} are as follows. 1) We prove in Thm.~\ref{thm:main} that the closed-loop equilibrium strategy of \eqref{eq:MDG} amounts to the mirror decent update in \eqref{eq:MP}, i.e., $u_i^*(t)=  -\nabla_i \psi_i (\nabla\phi_i^*(x_i), \nabla\phi_{-i}^*(x_{-i}))$, extending the Brezi-Ekeland variational principle \cite{brezis1976principe,tzen2023variational} in gradient flow to multi-agent mirror descent dynamics and differential games. 2) Since the finite-time learning trajectory $\{x_i(t), 0\leq t \leq T\}$ of \eqref{eq:MP} is equivalent to the equilibrium path of \eqref{eq:MDG}, we provide in Thm.~\ref{thm:finite-time-lyapunov} a finite-time quantification of MP's iterates using the  Lyapunov argument in MDG. 3) We further extend the variational principle to stochastic mirror play (SMP), often arising from machine learning applications due to noisy gradient evaluations. Akin to the deterministic case, we establish in Thm.~\ref{thm:stochasticcase} a variational equivalence between SMP and the stochastic mirror differential game (SMDG). In summary, the proposed dynamical system viewpoint gives us more theoretical tools to understand the stability and optimality of the learning scheme, which demonstrates the potential of our proposed variational interpretation of MP.

\subsubsection*{\bf Related work}
The long-run behavior of MP in discrete-time settings has been studied in the context of continuous convex games \cite{zhou2017mirror}, 
finite games \cite{mertikopoulos2016learning,mertikopoulos2019learning},
and potential games \cite{pan2023stochastic,pan2023resilience}, to name a few.
Specifically, when the game is strictly monotone, it converges to Nash equilibria under a second-order criterion called \textit{variational stability} \cite{mertikopoulos2019learning,gao2020continuous, shutian23erm}.

Such convergence behavior has also been studied under a continuous-time setting, such as in \cite{gao2020continuous,gao2022continuous},  which essentially relies on Lyapunov analysis.
Despite the recent progress, little has been discussed on how MP trajectory could be interpreted from a control-theoretic perspective. 
Inspired by \cite{tzen2023variational,brezis1976principe}, we explore this option by formulating a class of differential games that induces the equivalent MP dynamics.


\subsubsection*{\bf Preliminaries}
For a convex set $\mathcal{X}$ in a finite-dimensional Euclidean space $\mathscr{E}= (\mathbb{R}^n , \|\cdot\|)$, we denote $\mathcal{NC}_{\mathcal{X}} (x) :=  \{v \in \mathbb{R}^{n} \mid \langle v, y-x \rangle \geq 0 \quad \forall y \in \mathcal{X}\}$ as its normal cone at $x\in \mathcal{X}$.
A function $f: \mathscr{E} \mapsto  (-\infty, +\infty] $ is proper if $\exists x \in \mathbb{R}^n$ such that $f(x) < \infty$. 
Suppose $f$ is convex, lower semicontinuous, and proper, then it is said to be 1) closed if its epigraph is closed; 2) \textit{essentially smooth} if $f$ is differentiable on $\mathrm{int}(\mathrm{dom} (f) )\neq \emptyset$, and $\lim_{n \to \infty }\| \nabla f(x_n) \| \to + \infty $ whenever
$x_n$ converge to some boundary points of $\mathrm{dom}(f) $; 3) \textit{essentially strictly convex} if $f$ is strictly convex on every convex subset of $ \mathrm{dom}( \partial f)$; 4) \textit{Legendre} if it is both essentially smooth and essentially strictly convex.

Given a function $f$, its Fenchel conjugate $f^*: \mathscr{E}^* \to (-\infty, + \infty]$ is defined by $f^* (y) =  \sup_{x \in \mathscr{E}} \{ \langle x, y \rangle - f(x)\}$, where $\mathscr{E}^*$ is the dual space. $f^*$ is convex and closed if $f$ is proper. 
Given $x\in \mathscr{E}$ and $y \in \mathscr{E}^*$, the Fenchel coupling is given by $\mathcal{FC}_f(x, y) = f(x) + f^*(y) - \langle x, y \rangle$. By Fenchel-Young inequality, $\mathcal{FC}(x,y) \geq 0$ with equality holds if and only if $y \in \partial f (x)$ for proper and convex $f$, or $x \in \partial f^* (y)$ given that $f^*$ is closed. Given $f$ being proper, closed, convex, and differentiable over $\mathrm{dom}(\partial f)$, we can define the Bregman divergence from point $y \in \mathrm{dom}(\partial f)$ to $ x \in \mathrm{dom}(f)$ as $D_f ( x, y) :=  f(x) - f(y) - \langle \nabla f(y), x - y \rangle $. 
Given $f$ being Legendre, its gradient map $\nabla f: \mathscr{E} \mapsto \mathscr{E}^*$ is bijective, with $(\nabla f)^{-1} = \nabla f^*$. 
As a consequence, the Bregman divergence $D_{f^*}(x^\prime, y^\prime) =D_{f}(y, x) $ if $x^\prime = \nabla f(x)$ and $y^\prime = \nabla f(y)$. 
An operator $\Psi (\cdot): \mathcal{X}  \mapsto \mathbb{R}^n$ is said to be 1) monotone if $\langle \Psi(x) - \Psi (x^\prime) , x - x^\prime \rangle  \geq 0 $ $\forall x, x^\prime \in \mathcal{X}$; 2) $\mu$-strongly monotone w.r.t. $D_f(\cdot, \cdot)$ if $\langle \Psi(x) - \Psi (x^\prime) , x - x^\prime \rangle  \geq \mu D_f(x, x^{\prime})$ $\forall x, x^\prime \in \mathcal{X}$ for some $\mu > 0$ and Legendre function $f$.

\section{Problem Formulation}

\subsection{Monotone Game}
\label{subsec:monoton}
A monotone game is a continuous game given by the tuple $\mathcal{G} = (\mathcal{N}, \{ \mathcal{Y}_i \}_{i \in \mathcal{N}}, \{ \psi_i \}_{i \in \mathcal{N}})$, where $\mathcal{N}$ is the set of players, $ \mathcal{Y}_i \subseteq \mathbb{R}^{n_i}$, $n_i$ are the dimensions for the strategy spaces of player $i$. We denote $\mathcal{Y}:=  \prod_{i \in \mathcal{N}}  \mathcal{Y}_i$ as the joint strategy space, $\mathcal{Y} \subseteq \prod_{i \in \mathcal{N}} \mathbb{R}^{n_i} =  \mathbb{R}^n$, $n = \sum_{i \in \mathcal{N}} n_i $.
Each player is equipped with their own real-valued cost functional $\psi_i : \mathcal{Y} \to \mathbb{R}$, $y \mapsto \psi_i(y)$, where $y = (y_i)_{i \in \mathcal{N}}$ is the strategy profile for all players, with $y_i \in \mathcal{Y}_i$ being the strategy of player $i$, $y_{-i} \in \mathcal{Y}_{-i}$ as the strategies of the rest.

For each player $i$, $\mathcal{Y}_i$ is nonempty, closed, and convex; $\psi_i $ is jointly continuous in $y = (y_i, y_{-i})$; for all $y_{-i} \in \mathcal{Y}_{-i}$, $\psi_i (\cdot; y_{-i})$ is a convex, continuously differentiable function in the player's strategy variable $y_i$.
Given the conditions above, the game $\mathcal{G}$ admits a pure strategy \textit{Nash equilibrium} \cite{bacsar1998dynamic} denoted by $\overline{y}$. The equilibrium strategy profile $\overline{y} = (\overline{y}_i )_{i \in \mathcal{N}}$ satisfies the following condition: $\psi_i (\overline{y}_i; \overline{y}_{-i} ) \leq \psi_i ( y_i; \overline{y}_{-i})$, for all $i \in \mathcal{N} , y_i \in \mathcal{Y}_i$.

At the equilibrium, no player can benefit from unilaterally deviating from their equilibrium strategy. The pure strategy Nash equilibrium, by the characterization from \cite{facchinei2003finite} and \cite{friesz2010dynamic}, can also be equivalently formulated as the solution of a \textit{variational inequality} $\mathrm{VI}(\mathcal{Y}, \Psi)$: 
\begin{equation*}
     \langle y - \bar{y},  \Psi (\overline{y}) \rangle  \geq 0 \quad  \forall  y \in \mathcal{Y},  
\end{equation*}
with  $\Psi (y) :=  ( \nabla_i \psi_i ( y_i, y_{-i }))_{i \in \mathcal{N}}$ being the \textit{pseudogradient}, which is the stacked vector of all the $\psi_i$'s partial gradient at $x$, with respect to their own strategy variable. 
The solution to the variational inequality $\overline{y}$ also admits that the pseudogradient at $\overline{y}$ is included by its normal cone: $ \Psi(\overline{y}) \in \mathcal{NC}_{\mathcal{Y}} (\bar{y})$.
The game $\mathcal{G}$ is said to be monotone if its pesudogradient $\Psi$ is a monotone operator.

\subsection{Mirror Play as a Dynamical System}

Consider the scenarios where players are repeatedly interacting in the game in a prescribed continuous time duration $t \in [0, T]$, $T \in \mathbb{R}_+$, aiming to learn the Nash equilibrium gradually. 
The learning updates are considered to happen between infinitesimal steps, as in \cite{mertikopoulos2017convergence}. 
The players  are endowed with different \textit{mirror maps} $\phi_i: \mathbb{R}^{n_i} \to \mathbb{R} \bigcup \{ +\infty\}$, which are assumed to be of Legendre type.
The \textit{aggregated mirror map} $\phi: \mathbb{R}^{n} \to \mathbb{R}  \bigcup \{ +\infty\}$ is then given by $\phi( y ) =  \sum_{i \in \mathcal{N}} \phi_i (y_i)$. 
At each time $t$, each player $i$ has access to his own partial gradient $\nabla_i \psi_i$ given all the players' strategy profile $y$. 
The mirror play (MP) dynamics require the players to map their own strategy, from the \textit{primal space} $\mathcal{Y}_i$, to the \textit{dual space} $\mathcal{X}_i := \mathbb{R}^{n_i}$, and map the information from the dual space back to the primal space.

Let the \textit{dual trajectory space} be $\overline{\mathcal{X}}  = \prod_{i \in \mathcal{N}}\overline{\mathcal{X}}_i$, whose elements are smooth trajectory functions $\{ x_i(t) \in \mathbb{R}^{n_i}, 0 \leq t \leq T\} $; let the \textit{dual control spaces} be $\mathcal{U} = \prod_{i \in \mathcal{N}} \mathcal{U}_i$  whose elements are smooth control functions $\{ u_i (t) \in \mathbb{R}^{n_i}, 0 \leq t \leq T\}$; denote $f_i(\cdot) : \mathbb{R}^{n_i} \times  \mathbb{R}^{n_1} \times \cdots \times \mathbb{R}^{n_N}\to \mathbb{R}^n$ as the players' updating/aggregation rules. Then, let the state each player's learning process can be written as a dynamical system with their \textit{states} being their \textit{dual variables}:
\begin{equation}
\label{eq:mirror-play-dynamics}
\begin{aligned}
     \dot{x}_i (t) & =   f_i ( x_i (t),  u_i (t), u_{-i}(t) ), \quad x_i(0) = x_{i,0}  \\ 
      y_i & =  \nabla \phi^*_i (x_i)  \quad \forall i \in \mathcal{N}, 
\end{aligned}
\end{equation} 
with $ (x_{i,0})_{i \in \mathcal{N}}$ being a priori.
Since $\nabla \phi^*_i $ is bijective, ($\phi_i$ is Legendre), the observation of $y_i$ and $x_i$ are interchangeable.
There is a plethora of information patterns to be considered: e.g., the open-loop information structure, i.e., each player only has access to the initial state $\{x_{i,0}\}_{i \in \mathcal{N}}$, hence the \textit{policy class} $\Gamma_i$ only contains constant functions $\{\gamma_i(t, x_{i,0}) \}$; or closed-loop information structure, i.e., at time $t \in [0,T]$, all players have access to $\{ x_i(s),  0 \leq s \leq t\}_{i \in \mathcal{N}}$. See \cite{bacsar1998dynamic, tao_info} for a more nuanced discussion about information structures.

Here, to better capture the MP learning framework, we only consider the closed-loop information structure, under which we prespecify $\Gamma_i$ to consist of measurable mappings $\gamma_i : [0,T] \times \overline{\mathcal{X}} \to \mathcal{U}_i$. 

\subsection{From Mirror Play to Mirror Differential Game}
\label{subsec:MP2MDG}
One archetype of mirror play is mirror descent (MD), which originally takes the following form, 
$ f_i(x_i, u_i (t), u_{-i}(t)) \triangleq  u_i (t), $
with the $u_i (t) = \nabla_i \psi_{i} (x_i(t), x_{-i}(t))$ being a feedback strategy determined by partial gradients with respect to dual states.
Other variants, as discussed in \cite{gao2020continuous}, might take consideration of a discounted factor $\eta$, as well as an exponential discounted aggregation, i.e., $f_i(x_i, u_i (t), u_{-i}(t)) \triangleq  \eta ( - x_i (t) + u_i (t)).$

Let the initial state $x_0 := [x^{\top}_{1,0}, \ldots, x^{\top}_{N,0}]^{\top}$, and $x(t): = [x_1^{\top}(t),  \ldots, x_N(t)]^{\top}, 0 \leq t \leq T$ be the stacked version of dual trajectories, we can write the individual differential equations into a more compact form \cite{mertikopoulos2019learning,mertikopoulos2017convergence}.
\begin{equation} \label{eq:diffcompact}
\begin{aligned}
         \dot{x} (t)  = u (t) , \ x (0) = x_0 ,  \quad y(t) = \Phi^*( x(t)),
\end{aligned}
\end{equation}
where $\Phi^*$ and $u$ are the stacked vectors of $\nabla \phi^*_i$ and $u_i$, respectively. 
Apparently, $u(t)$ is continuous in $t$ and uniformly Lipschitz in $x(\cdot)$, when the admissible controls $\gamma_i(t,x)$ are continuous in $t$ and uniformly Lipschitz in $x(\cdot)$.
Hence, the differential equation \eqref{eq:diffcompact} admits a unique continuous dual state trajectory for any admissible $\gamma_i$ such that $u_i (t) = \gamma_i(t,x)$.


The goal of this paper is to provide a variational interpretation of MP, by lifting the monotone game to the corresponding differential game. 
To this end, we define the individual cost functionals for the differential game. 
Let $\psi_i^*$ to be the Fenchel conjugate of function $y_i \mapsto \psi_i( y_i, y_{-i})$ given all the opponent strategy data $y_{-i}$, i.e., $\psi_i^*( u_i \mid y_{-i}) =  \sup_{y_i \in \mathcal{Y}_i} \left\{  \langle u_i, y_i \rangle -   \psi_i (y_i, y_{-i})   \right\}$. 
We equip with each player $i \in \mathcal{N}$ a  cumulative cost $J_i$:
\begin{equation}\label{eq:cumucost}
     J_i (x_0, \{u_i\}_{i \in \mathcal{N}}) = \int_{t=0}^{T} c_i (x(t), u_i(t), u_{-i}(t)) \mathrm{d}t + q_i(x(T)), 
\end{equation}
where the terminal cost is defined as the Bregman divergence from the equilibrium dual state to the terminal dual state:
\begin{equation*}
    q_i(x(T)) := D_{\phi^*_i} (x_i(T),  \overline{x}_i) = D_{\phi_i } (y_i(T),\overline{y}_i), 
\end{equation*}
and the stage cost is defined to be analogous to the Fenchel coupling between $x$ and $u_i$:
\begin{align*}
     c_i(x, u) :=& \psi_i(\nabla \phi_i^* (x_i), \nabla \phi_{-i}^* (x_{-i}))\\
     &+  \psi_i^*( -u_i \mid \nabla\phi_{-i}^* (x_{-i})) + \langle  u_i, \bar{y}_i \rangle,
\end{align*}
which resembles to the Brezis-Ekeland action functional \cite{brezis1976principe,tzen2023variational}, is continuously differentiable and convex in both $x_i$ and $u_i$.
This formulation implies that during the MP, as the dual state approaches equilibrium, both the stage cost and the cumulative cost diminish correspondingly. This property has been discussed in \cite{zhou2017mirror} as ``Fenchel conforming''.
Thus, the underlying differential game ( mirror differential game) exhibits a profound linkage with the monotone game $\mathcal{G}$.

To be distinguished from the Nash equilibrium (NE) of $\mathcal{G}$, the Nash equilibrium strategy set $\{\gamma_i^*\}_{i \in \mathcal{N}}$ under closed loop information structure is referred to as \textit{closed-loop-equilibrium policy} (CLE). 
The subsequent will be dedicated to showing that during the MP learning process of game $\mathcal{G}$, its non-equilibrium \textit{mirror path}, i.e., the finite-time MP iterates from the primal and dual space $\{x(t), y(t), 0\leq t \leq T\}$, follows a variational principle associated with the underlying mirror differential game (MDG).

\section{mirror path Characterization}

\subsection{The Variational Principle}
Define the Hamiltonian $H_i: [0,T] \times \mathbb{R}^n \times \mathbb{R}^n \times \mathbb{R}^n \to \mathbb{R}$:
\begin{equation*}
    H_i ( t,  p_i, x, u) \triangleq c_i (x(t), u(t)) + \langle p_i(t), u (t) \rangle,
\end{equation*}
where $p_i: [0,T] \to \mathbb{R}^n$ are costate functions to be specified later.
Now, from optimal control theory, every individual essentially aims to, if in particular, a CLE policy $\gamma^*$ exists, then the value functions, $V_i(t, x) \triangleq J_i(t, x, \gamma^*) =  \int_{\tau=t}^{T} c_i (x(\tau), \gamma^*(\tau, x(\tau))) \mathrm{d}\tau + q_i(x(T))$, must also exist, and can be found through solving the Hamilton-Jacobian-Bellman (HJB) equations. 

\begin{prop} 
 Under the closed-loop information structure, the CLE policy set $\{\gamma^*_i (t, x) = u_i^* (t) ; \ i \in \mathcal{N} \}$ generates a primal-dual  path $\{(x^*(t), y^*(t) ),  0 \leq t \leq T\}$ if there exist functions $V_i: [0,T] \times \mathbb{R}^n \to \mathbb{R}$ that satisfies the   HJB equations
  \begin{equation}\label{eq:hjbi}
  \begin{aligned}
        - \partial_t V_i (t,x) & =  \min_{u_i \in \mathcal{U}_i} H_i(t, \nabla_x V_i,  x, u_i, \gamma^*_{-i}),   \\
       =  \langle & \nabla_x   V_i(t, x), (\gamma^*_i, \gamma^*_{-i}) \rangle + c_i(x, \gamma^*_i, \gamma^*_{-i})    \\
      V_i(T,x)   & = q_i (x), \quad i \in \mathcal{N}.
  \end{aligned}
 \end{equation}
 Furthermore, the primal-dual path satisfies the following relations:
 \begin{equation}\label{eq:noneqpath}
 \begin{aligned}
       \dot{x}^*(t) & =  u^*(t), y^*(t) = \Phi^* (x^*(t)), \ x^*(0) = x_0 , \\
    \gamma^*_i(t, x) \equiv u_i^*(t) & = \argmin_{u_i \in \mathcal{U}_i} H_i (t, p_i, x^*, u_i, u^*_{-i}),  \\
     \dot{p}_i (t) & = - \nabla_x H_i (t, p_i, x^*, u^*) , \\ 
      p_i(T) & = \nabla_x q_i(x^*(T)) . 
 \end{aligned}
 \end{equation}
The corresponding cumulative cost for $i$ is $V_i(0, x_0)$.
\end{prop}

So far, we already know the boundary conditions $V_i(T,x(T)) = D_{\phi^*_i} (x_i(T), \overline{x}_i)$. We might as well just assume that 
\begin{equation}\label{eq:valuefunc}
    V_i (t, x)= V_i (x(t)) =  D_{\phi^*_i} ( x_i (t), \overline{x}_i).
\end{equation} Then, Lemma \ref{lemma:nonneg} indicates that $V_i$ indeed solves the HJB equations, further implying that the primal-dual path by the CLE policy corresponds to the finite-time mirror path of the mirror play dynamics in \eqref{eq:mirror-play-dynamics}, as stated in Thm.~\ref{thm:main}.  

\begin{lemma}\label{lemma:nonneg}
    The following inequality holds for the value functions $V_i$ defined by \eqref{eq:valuefunc}, $i \in \mathcal{N}$: 
    \begin{equation}
         \langle  \nabla_x   V_i(t, x), (u_i, \gamma^*_{-i}) \rangle + c_i(x, u_i, \gamma^*_{-i})  \geq 0,
    \end{equation}
    with the equality attained if and only if  $u_i =  - \nabla_i \psi_i (\nabla \phi_i^* (x_i), \nabla \phi_{-i}^* (x_{-i}))$ is the MD policy.
\end{lemma}

\begin{proof}
   Since $V_i$ does not depend on $x_{-i}$, $\nabla_x V_i (t, x) $ has $0$ entries on the dimensions corresponding to those of $x_{-i}$, i.e., $$
     \langle \nabla_x   V_i(t, x), (u_i, \gamma^*_{-i}) \rangle =  \langle \nabla \phi^*_i (x_i) -  \nabla\phi_i^* (\overline{ x}_i ) ,  u_i \rangle .
   $$
   Then, due to the fact that $\overline{y}_i = \nabla\phi^*_i(\overline{x}_i)$, and $y_i = \nabla\phi^*_i(x_i) $ simple calculation yields,
   \begin{align*}
       & \quad \langle  \nabla_x   V_i(t, x), (u_i, \gamma^*_{-i}) \rangle + c_i(x, u_i, \gamma^*_{-i}) \\
        & = \psi_i(\nabla \phi_i^* (x_i), \nabla \phi_{-i}^* (x_{-i})) +  \psi_i^*( -u_i \mid \nabla \phi_{-i}^* (x_{-i})) + \langle  u_i, \bar{y}_i \rangle \\ & \quad +   \langle  \nabla \phi_i^* (x_i) -  \overline{y}_i,  u_i \rangle  \\
        & = \mathcal{FC}_{\psi_i(\cdot, \nabla \phi_{-i}^* (x_{-i}))}(  y_i ,  - u_i ) \geq 0 ,
   \end{align*}
   where the equality can be attained if and only if $-u_i =  \nabla_i \psi_i( y_i, \nabla \phi_{-i}^* (x_{-i}))$.
   
\end{proof}

\begin{theorem}\label{thm:main}

Consider the MDG defined by \eqref{eq:diffcompact} and \eqref{eq:cumucost},  we have the following for each $i \in \mathcal{N}$: 
    \begin{itemize}
        \item[i)] For any $t\in [0,T]$, the following holds:
        \begin{equation}
              \int_{0}^{t} c_{i}(x(\tau), u(\tau)) d \tau \geq V_{i}(x_0)-V_{i}(x(t)) .
        \end{equation}
        Further, for any policy $u \in \mathcal{U}$, $J_i(x_0, u_i, u_{-i}) \geq V_i (x_0)$.
        \item[ii)] For any priori $x_0$, the primal-dual path $\{(x^*(t), y^* (t)); \ 0 \leq t \leq T \}$ is generated by the CLE policy set $\{\gamma^*_i (t, x) = u_i^* (t) =  - \nabla_i \psi_i( \nabla\phi^*_i(x_i), \nabla \phi_{-i}^* (x_{-i}))  ; \ i \in \mathcal{N}  \}$, which satisfies the relations \eqref{eq:noneqpath}, where the costate dynamics follows:
        \begin{equation*}
             \begin{aligned}
                 \dot{p}_i(t) & =  -  [ 0, \ldots, (\nabla_i \psi_i (y))^{\top}, \ldots, 0]^{\top} \\
                  p_i(T) & = [ 0, \ldots, (\nabla_i \psi_i (y) - \nabla_i \psi_i (\overline{y}))^{\top}, \ldots, 0]^{\top},
             \end{aligned}
        \end{equation*}
        and the cumulative cost $J_i(x_0, u^*_i, u^*_{-i}) = V_i(x_0)$.
        
        \item[iii)] When the time horizon $T\to+\infty$, then, the closed-loop system $\dot{x}^*(t) = u^*(t)$ is asymptotically stable, with $V(x) := \sum_{i \in \mathcal{N}} V_i (x(t))$ being a global Lyapunov function.  
    \end{itemize}
\end{theorem}

\begin{proof}
i): Consider an arbitrary policy profile $u \in \mathcal{U}$. By standard analysis:
    \begin{align*}
        V_i (x (t)) - V_i (x_0) & = \int_{0}^t \frac{\mathrm{d}}{\mathrm{d}s} V_i(x(s))ds \\
        =  \int_{0}^t  \langle  \nabla_x   V_i(t, x), (u_i, u_{-i}) \rangle
        & \geq -\int_{0}^t c_i (x(s), u(s)) ds, 
    \end{align*}
    the last inequality is due to similar calculations in the proof of Lemma \ref{lemma:nonneg}, rearranging the terms yields the inequality. 
    Hence,  
    $J_i (x_0, u_i, u_{-i}) = \int_{0}^T c_i (x(s), u(s)) ds + V_i(x(T)) \geq V_i(x_0)$. 

ii): By Lemma \ref{lemma:nonneg}, for all $i \in \mathcal{N}$, when the only minimizers $\gamma_i^*$ are attained, $V_i $ satisfy the HJB equation \eqref{eq:hjbi}.
Thus, the costate dynamics can be obtained as follows: 
\begin{align*}
    \dot{p}_i(t) & =  -\nabla_x V_i (x(t)) \\
     p_i(T) & = \nabla_x V_i(x(T)) \quad   i \in \mathcal{N}.
\end{align*}
When $\gamma^*_i$ is attained, the equality in i) is attained as well, thus $J_i(x_0, u^*_i, u^*_{-i}) = V_i(x_0)$. 

iii): 
The mirror path is generated by the closed-loop system  $\dot{x}(t) = (- \nabla_i \psi_i (\nabla \phi_i^* (x_i), \nabla \phi_{-i}^* (x_{-i})) )_{i \in \mathcal{N}}$. 
Now, we examine the function $V(x)$, which satisfies:
\begin{itemize}
    \item $V(\overline{x}) = \sum_{i \in \mathcal{N}} V_i(\overline{x}_i) = 0$; 
    \item $V(x)$ is positive definite for $x \neq \bar{x}$, twice continuously differentiable and strictly convex due to the definition of Bregman divergence; 
    \item coercive due to the fact that  $\phi^*_i$ are essentially strictly convex. 
\end{itemize} 
Along the path, by Lemma \ref{lemma:nonneg} we have 
\begin{equation}\label{eq:c+dv}
    c_i (x, u^*) + \frac{\mathrm{d}}{\mathrm{d}t}V_i (x) = 0 . 
\end{equation} 
Thus, 
\begin{align*}
     & \quad \quad  \frac{\mathrm{d}}{\mathrm{d}t} V (x(t))  =  \dot{V} (x(t), - \nabla \psi (y(t))) \\
              & = - \sum_{i \in \mathcal{N}} c_i(x(t), - \nabla_i \psi_i (y(t))) \\
              &  =  \sum_{i \in \mathcal{N}}   \langle \nabla_i \psi_i (y(t))), \overline{y}_i\rangle - \psi_i^*(\nabla_i \psi_i (y(t)) ) - \psi_i ( y(t))  \\ & 
            =  \sum_{i\in \mathcal{N}}  - \mathcal{FC}_{\psi(\cdot, y_{-i})} ( y_i, \nabla_i \psi_i(y))+ \langle \nabla_i \psi_i(\overline{y}), \overline{y}_i - y_i(t)\rangle 
            \\ & \leq  \sum_{i\in \mathcal{N}}  \langle \nabla_i \psi_i(\overline{y}), \overline{y}_i - y_i(t)\rangle \leq 0,
\end{align*}
 where the second last inequality is due to Fenchel-Young inequality and the last inequality is due to the monotonicity of game $\mathcal{G}$ and the definition of $\overline{y}$.
    
\end{proof}

\subsection{Finite-time Quantification}

Now we are ready to quantify the non-asymptotic ``efficiency'' of the mirror path from the variational perspective. 
The convergence analysis of MP-based learning in games has been studied before, such as in \cite{krichene2015convergence} for population games in discrete time;  
and a similar dynamical-system-based approach can be found in \cite{gao2020continuous}. 
Here, we aim to show the potential for variational interpretation to serve as a general framework to obtain performance guarantees.

\begin{theorem}
\label{thm:finite-time-lyapunov}
    Let $\{(x(t), y(t)); \ 0 \leq t \leq T \}$ be generated by Mirror play \eqref{eq:mirror-play-dynamics} (equivalently, the CLE policy of MDG) $\{u^*(t) ; 0 \leq t \leq T\}$. Then, it holds that 
    \begin{equation}\label{eq:monotonconv}
      \langle \Psi(y),  \bar{y}_{i, [0,T]} - \bar{y}  \rangle \leq \frac{1}{T} \sum_{i \in \mathcal{N}} D_{\phi_i} (\overline{y}_i, y_{i,0}), 
    \end{equation}
    which indicates that $ \bar{y}_{i, [0,T]}$ converges to some $y$ such that $\Psi(y) \in \mathcal{NC}_{\mathcal{Y}}(y)$ as $T \to \infty$. 
    In addition, if $\mathcal{G}$ is $\mu$-strongly monotone with respect to $D_{\phi}(\cdot, \cdot)$ where $\phi$ is the aggregated mirror map, then, $V (x(t)) \leq e^{ - \mu t} V(x_0)$. 
    In this case, the closed-loop system is exponentially stable as $T \to \infty$.
\end{theorem}

\begin{proof}
   i):
   Using the fact \eqref{eq:c+dv}, we have that for each player $i \in \mathcal{N}$, the path stage cost satisfies
    \begin{align*}
       c_i (x, u^*) & = - \frac{\mathrm{d}}{\mathrm{d}t}V_i (x)  = - \frac{\mathrm{d}}{\mathrm{d}t} D_{\phi^*_i}(x_i(t), \bar{x}_i)\\ & = \langle \nabla_i \psi_i ( \nabla \phi^*_i (x(t))) , \nabla\phi^*_i(x_i(t))  - \nabla \phi^*_i(\overline{x}_i)  \rangle  \\
     & = \langle  \nabla_i \psi_i (y(t)), y_i (t) - \overline{y}_i\rangle  
    \\ & = \langle  \nabla_i \psi_i (\overline{y}), y_i (t) - \overline{y}_i\rangle \\ & \quad  + \langle  \nabla_i \psi_i (y(t)) -  \nabla_i \psi_i (\overline{y}), y_i (t) - \overline{y}_i\rangle 
    .
    \end{align*}
    where we use the convexity of $\psi_i (\cdot, y_{-i})$, 
   Due to Theorem \ref{thm:main}, we have for each $i \in \mathcal{N}$,
    \begin{align*}
          V (x_0) & =   \sum_{i \in \mathcal{N}} D_{\phi^*_i} (x_i(T), \overline{x}_i )  +  \int_{t=0}^T  \sum_{i \in \mathcal{N}} c_i (x(t), u^*(t)) \mathrm{d}t 
          \\ & \geq   \int_{t=0}^T  \sum_{i \in \mathcal{N}} \langle  \nabla_i \psi_i (\overline{y}), y_i (t) - \overline{y}_i\rangle \mathrm{d}t ,
    \end{align*}
    where the inequality is due to the nonnegativity of Bregman divergence and the monotonicity of $\mathcal{G}$.
    
    Then, let $\bar{y}_{i, [0,T]} = \frac{1}{T} \int_{t=0}^T y_i(t) \mathrm{d}t$, $i \in \mathcal{N}$ be the average strategy profile, we arrive at:
    \begin{align*}
          & \quad \sum_{ i \in \mathcal{N}} \langle  \nabla_i \psi_i (\overline{y}), \bar{y}_{i, [0,T]} - \overline{y}_i\rangle
          \\ & =  \frac{1}{T} \int_{t=0}^T  \sum_{i \in \mathcal{N}}\langle  \nabla_i \psi_i (\overline{y}), y_i (t) - \overline{y}_i\rangle \mathrm{d}t  \\
        & \leq  \frac{1}{T} \sum_{i \in \mathcal{N}} D_{\phi_i} (\overline{y}_i, y_{i,0}) .
    \end{align*}
    By the characterization of variational inequality, $\overline{y}$ satisfies that  $\Psi(\overline{y}) \in \mathcal{NC}_{\mathcal{Y}}(\overline{y})$, $ \langle \Psi(\overline{y}), y - \overline{y} \rangle  \geq 0$ for any $y \in \mathcal{Y}_i$. Thus, suppose that $\bar{y}_{i, [0,T]} \to  y^{\prime}$ such that $\Psi (y^{\prime}) \not\in \mathcal{NC}_{\mathcal{Y}}(y^{\prime})$, then, we have a contradiction.

 ii):   When the game is $\mu$-strongly monotone, 
 \begin{align*}
   & \quad \ \frac{\mathrm{d}}{\mathrm{d}t} V(x(t)) = \sum_{i \in \mathcal{N}}\frac{\mathrm{d}}{\mathrm{d}t}V_i(x_i(t)) \\ &
    = \sum_{i \in \mathcal{N}}  \langle  \nabla_i \psi_i(\overline{y}),\overline{y}_i - y_i (t) \rangle  \\ & \quad - \sum_{i \in \mathcal{N}}  \langle  \nabla_i \psi_i (y(t)) -  \nabla_i \psi_i (\overline{y}), y_i (t) - \overline{y}_i\rangle  
    \\ & \leq  - \sum_{i \in \mathcal{N}}  \langle  \nabla_i \psi_i (y(t)) -  \nabla_i \psi_i (\overline{y}), y_i (t) - \overline{y}_i\rangle 
    \\ & \leq  - \mu  D_{\phi} ( \overline{y} , y (t)) = - \mu \sum_{i \in \mathcal{N}} D_{\phi_i} (\overline{y}_i , y_i(t)), 
 \end{align*} 
 where the first inequality is due to variational stability, the second inequality is due to $\mu$-strongly monotonicity. 
 By Gronwall's inequality, one can obtain exponential stability:
  \begin{equation*}
       V (x(t)) \leq  \int_{0}^t -\mu V( x(\tau ))  \mathrm{d}\tau  \leq e^{ - \mu t} V(x_0) .
  \end{equation*}

\end{proof}

\subsection{Example: Cournot Duopoly}
Consider the following 2-person multi-dimensional Cournot game. The cost for firm $i =1, 2$ for producing $y_i \in \mathbb{R}^n$ units is $p_i^{\top} y_i$ where $p_i$ are the cost vectors for firm $i$. 
The total output of the firms is $Y = y_1 + y_2 \in \mathbb{R}^n$, the market price is $P =  \max \{M - Y , 0\} \in \mathbb{R}^n$, assuming $ p_i \succ   M \in \mathbb{R}^n$  for $i = 1,2$. Then, the profit $\psi_i(y_i, y_{-i})$ for firms $i=1,2$ are $ \psi_i(y_i, y_{-i}) : = - ( P - p_i )^{\top}y_{i}$. 

This game is a monotone game due to the monotonicity of $\Psi(y) := \mathrm{col}\left( ( M - 2 y_i  - y_{-i} - p_i )_{i = 1, 2} \right)$. The Nash equilibrium strategies are $\overline{y}_i = \frac{1}{3} ( M + p_{-i} - 2 p_i)$ for firms $i=1,2$, where the market price vector is $P = \frac{1}{3}(M + p_i + p_{-i})$.

Picking quadratic mirror maps $\phi_i (y_i) = \frac{1}{2} \langle y_i  , A_i y_i \rangle$ where $A_i$ are positive definite symmetric matrices in $\mathbb{R}^{n\times n}$, $\nabla \phi_i (y_i) =  A_i y_i $, $\nabla \phi^{-1} (x_i) = A_i^{-1}x_i  \equiv \nabla \phi^*_i (x_i)$. We have $\overline{x}_i  =  \frac{1}{3} \langle A_i,  M -2 p_i + p_{-i} \rangle$, $D_{\phi^*_i} (x_i, \bar{x}_i)= \| x_i - \bar{x}_i\|^2_{A_i^{-1}} + \frac{1}{2} \|\bar{x}_i\|^2_{A^{-1}_i} $. 
Hence we can create a general-sum linear-quadratic game via a similar procedure, where two firms are trying to minimize their own cost functional $J_i(x_0, u_1, u_2)$ starting from arbitrary dual state $x_0$. However, here we can anticipate the optimal strategies for both firms to be found through the mirror play without explicitly solving the coupled Ricatti equations.

\section{Stochastic Mirror Play Case}

 When each player $i \in \mathcal{N}$ only has access to a ``black box" model of $\psi_i$, e.g., a first-order \textit{orcacle}, the information may be noisy subject to 1) the measurement/transmission errors; 2) the fluctuation of cost functionals.

As a motivating example, consider Generative Adversarial networks (GAN). 
 Let $Z \sim  \boldsymbol{\mathcal{N}} (0, \mathbf{I})$ and $X \sim \boldsymbol{\mathcal{N}}(\mu,  \mathbf{I})$ be two $n$-dimensional Gaussian random vectors, with $\mu = \mathbb{E} (X)$. The goal of GAN is to learn a \textit{generator network} $G_{y_1} : \mathbb{R}^{n_1} \to \mathbb{R}^n$, and a \textit{discriminator network} $D_{y_2} : \mathbb{R}^{n_2} \to \mathbb{R}$ parameterized by ${y_1} \in \mathbb{R}^{n_1}$ and $y_2 \in \mathbb{R}^{n_2}$ respectively, such that the generator produces samples, $\mathbb{E}[G_{y_1} (Z)] = \mathbb{E} [ X ]$, which cannot be distinguished by the discriminator in average. Literature has shown that GAN essentially requires solving a zero-sum game \cite{daskalakis2017training}:  $\min _{{y_1} \in \mathbb{R}^{n_1}} \max _{y_2 \in \mathbb{R}^{n_2}} \mathbb{E}\left(D_{y_2}(X)\right)-\mathbb{E}\left(D_{y_2} \left(G_{{y_1}}(Z)\right)\right) .
$
However, since $\mu$ is unknown the players can only get noisy feedback from the distribution $\mathcal{N}(\mu, \mathbf{I})$ to evaluate the expectation.

Hence, we study stochastic mirror play within the \textit{stochastic differential game} framework. Let $(\Omega, \mathcal{F}, (\mathcal{F}_t)_{ 0 \leq t \leq T}, \mathbb{P})$ be the filtered probability space, and let $(W_t)_{t \geq 0}$ be the standard $n$-dimensional Wiener process adapted to the (right-continuous) filtration $(\mathcal{F}_t)_{ 0 \leq t \leq T}$.
The (stacked) MP dynamics under such an information pattern can be thus formulated as a \textit{stochastic differential equation} (SDE). 
\begin{equation*}
\begin{aligned}
      \mathrm{d} X_t & = \gamma (t, X_t) \mathrm{d}t + \sigma(X_t, t) \mathrm{d}W_t, \quad X_0 = x_0,  \\
      Y_t  & = \Phi^*(X_t), \quad \quad \quad t \in [0,T],
\end{aligned}
\end{equation*}
where $\sigma(t, \cdot): [0,T] \times  \mathcal{X} \to \mathbb{R}^{n \times n}$ is the volatility matrix.
When the admissible control $ \gamma (t, \cdot)$ and $\sigma (t, \cdot)$ are both Lipschitz-continuous in $x$, the SDE admits a unique strong solution, i.e., $X_t = x_0 + \int_{0}^t u(s, X_s) ds + \int_{0}^t \sigma (s, X_s) ds$ such that $\gamma (t, X_t) = \mathrm{col}(\gamma_i(t, X_t))_{i \in \mathcal{N}}$ are $\mathcal{F}_t$-measurable. Similarly, the expected cost-to-go functionals are:
\begin{equation*}
      J_i (t, x_0, u) := \mathbb{E}[\int_{s = t}^T c_i( X_s, u(s, X_s ) ) ds + q_i (X_T ) ] \quad i \in \mathcal{N}. 
\end{equation*}

We consider a setting where the volatility matrix is scaled by the inverse of diagonally stacked Hessians $\sigma \sigma^{\top} (t, X_t) :=  \sqrt{2 \varepsilon}\mathrm{diag} \left(\nabla^2 \phi_i^* ( X_{t,i} )  \right)^{-1}$, where $\varepsilon$ is a scaling factor. This setting resembles the scenario where every player follows the so-called \textit{Hessian Reimannian dynamics} \cite{alvarez2004hessian}. 
The diagonal matrix implies the independence of feedback noises between players.

\begin{prop}
     When the players only have access to noisy information feedback, 
     an $N$-tuple of feedback strategies $\{ \gamma^*_i; i \in \mathcal{N}\}$ provides a CLE solution if there exist suitably smooth functions $V_i : [0,T ] \times \mathbb{R}^n \to \mathbb{R}$, $ i \in \mathcal{N}$ satisfying the coupled semi-linear parabolic partial differential equations (HJB):
     \begin{equation}\label{eq:stochasticcasepde}
        \begin{aligned}
               - \frac{\partial V_i}{\partial t} & =  \min_{u_i \in \mathbb{R}^{n_i}} \left\{ \mathcal{L}^{(u_i, \gamma_{-i}^*)} V_i (t, x) + c_i(x, u_i, \gamma^*_{-i}) \right\}
              \\ V_i(T, x) & = q_i ( x ),  \quad \quad i \in \mathcal{N}  .
        \end{aligned}
     \end{equation}
     where $\mathcal{L}^{u}$ are second-order partial differential operators with respect to $u \in \mathbb{R}^{n}$, i.e., for any twice differentiable function $f$ in $\mathbb{R}^{n}$, $
           [\mathcal{L}^{u} f ] \triangleq  \frac{1}{2} \mathrm{trace}( \sigma \sigma^{\top} \nabla^2 f)  + \langle \nabla f, u \rangle$.
\end{prop}

We are still bound to construct value functions such that the boundary conditions $V_i ( T, x) = D_{\phi^*_i}( x_i, \overline{x}_i) $ are satisfied, with a small difference in that the cost to go at $t$ shall also consist of the quadratic variation.

\begin{theorem} \label{thm:stochasticcase}
  Suppose the partial gradients $\nabla_i \psi_i$ and Hessians $\nabla^2 \phi^*_i$ are Lipschitz-continuous in $x_i$, the following value functions 
     \begin{equation}
         V_i( t, x ) =  D_{\phi^*_i} (x_i (t), \overline{x}_i ) + \varepsilon n_i (T - t)
     \end{equation}
     satisfy the HJB equations \eqref{eq:stochasticcasepde}. 
     The set of feedback strategies $\{ \gamma^*_i (t,x ) \equiv u^*_i(t) =  - \nabla_i \psi_i (x_i, x_{-i}) \}$ provides a Nash equilibrium solution to the differential game. 
\end{theorem}

\begin{proof}
 By speculation, it can be verified that
\begin{align*}
     \frac{\partial V_i (t,x)}{ \partial t} & =  \varepsilon n_i,  \\
    \nabla V_i (t,x) & =  [  0 , \ldots, \underbrace{ (\nabla \phi^*_i (x_i) - \nabla \phi^*_i (\overline{x}_i))^{\top} }_{\text{the }i' \text{th stack}},  \ldots , 0]^{\top} ,
    \end{align*}
    and \begin{align*}
     \nabla^2 V_i (t,x) & = \begin{pmatrix}
        0 & \cdots & 0 & \cdots & 0 \\
\vdots & \ddots & \vdots & & \vdots \\
0 & \cdots & \underbrace{\nabla^2 \phi^*_i (x_i) }_{ \text{the } i \times i' \text{th block}}& \cdots & 0 \\
\vdots &  & \vdots & \ddots & \vdots \\
0 & \cdots & 0 & \cdots & 0
     \end{pmatrix} .  
\end{align*}
Thus, the summation   
\begin{align*}
    & \quad  \frac{\partial V_i (t,x)}{ \partial t}+  \mathcal{L}^{u} V_i (t,x) + c_i (x,u) \\
    & = - \varepsilon n_i  +  (\nabla \phi^*_i (x_i)  - \nabla \phi^*_i (\overline{x}_i ))^{\top} u_i  \\
    & + \mathrm{trace}((\nabla^2 \phi^*_i (x_i))^{-1}\nabla^2 \phi^*_i (x_i)) + c_i(x, u ) \\
     &  =  \psi_i( \nabla \phi^*_i(x_i), \nabla \phi^*_{-i}(x_{-i})) + \psi^*_i(-u_i) + \langle \nabla \phi^*_i(\overline{x}_i), u_i \rangle, 
\end{align*}
 which is independent from $u_{-i}$, is minimized if.f. $ - u^*_i =  \nabla_i \psi_i (x_i, x_{-i})$ according to Lemma \ref{lemma:nonneg}. 
 Hence, by the assumption on $\nabla_i \psi_i$ and $\nabla^2 \phi^*_i$, the  SDE admits a unique strong solution and $u^*_i$ is a valid admissible control.
\end{proof}

We can also show similar finite-time quantification results on the variational stability through a similar procedure from the proof of Theorem \ref{thm:main}. 
\begin{prop}
    Let $\{(X_t, Y_t); 0 \leq t \leq T\}$ be generated by the CLE policy $\{ u^*(t); 0 \leq t \leq T\}$ with boundary condition $(x_0, y_0)$. Then, it holds that 
    \begin{equation}\label{eq:stoquan}
         \mathbb{E} \left[ \langle  \Psi (\overline{y}),  \tilde{Y}_T - \overline{y} \rangle \right] \leq \frac{1}{T} \sum_{ i \in \mathcal{N}} D_{\phi^*_i} (x_{0,i}, \overline{x}_i) +  \varepsilon n ,
    \end{equation}
    where $\tilde{Y}_T = \int_{t= 0}^T Y_t  \mathrm{d}t$ is a time average.
    Further, when the game is $\mu$-strongly monotone w.r.t. $D_{\phi}$, it holds that
    \begin{equation}
        \mathbb{E} \left[ D_{\phi} ( \overline{y}, Y_t ) \right] \leq  e^{-\mu t} D_{\phi} (\overline{y}, y_0) + (1 - e^{-\mu t}) \frac{ \varepsilon n }{\mu}. 
    \end{equation}
\end{prop}

\begin{customproof}{Proof Sketch}
Let $V(t,x) = \sum_{i \in \mathcal{N}} V_i (t, x )$,
a similar calculation yields 
\begin{equation*}
    \begin{aligned}
        V(0, x_0)& = \varepsilon n T +  \sum_{ i \in \mathcal{N}} D_{\phi^*_i} (x_{0,i}, \overline{x}_i)  \\
        & = \mathbb{E} \left[ \int_{0}^T c_i(X_t, u^*_t)\mathrm{d}t + \sum_{i\in\mathcal{N}}D_{\phi^*_i}(X_{i,T}, \overline{x}_i)  \right]  \\
        & \geq \mathbb{E} \left[ \int_{t=0}^T \sum_{ i \in \mathcal{N}} \langle \nabla_i \psi_i (\overline{y}), Y_{i,t} - \overline{y}_i \rangle \mathrm{d}t \right].
    \end{aligned}
\end{equation*}
Dividing both sides by $T$ and applying Fubini's theorem gives \eqref{eq:stoquan}. 
Since $D_{\phi}(\overline{y}, Y_t) = \sum_{i \in \mathcal{N}} D_{\phi^*_i} (X_{i,t}, \overline{x}_i)$ is a $\mathcal{F}_t$-adapted Ito's process, leveraging the calculation from the proof of Theorem \ref{thm:main}: 
\begin{align*}
    & \quad d D_{\phi} (\overline{y}, Y_t) \leq  \underbrace{\sum_{i \in \mathcal{N}} \langle \nabla_i \psi_i (\overline{y}) ,  Y_{i,t} - \overline{y}_{i} \rangle \mathrm{d}t}_{\leq 0 \text{ otherwise reach Nash equilibrium}} \\
    & + \frac{1}{2}\mathrm{trace} ( 2 \varepsilon \mathrm{diag} (\nabla^2 V (\nabla^2\phi^*_i(X_t))^{-1}) \mathrm{d}t \quad ( = \varepsilon n)\\
     & - \mu D_{\phi}(\overline{y} , Y_{t})  \mathrm{d}t + 
 \underbrace{  \left(\nabla V \right)^{\top} \sigma (t, X_t) \mathrm{d}W_t }_{ \text{ constitutes Ito's integral}}, 
\end{align*}
where we use Ito's formula in differential form.
Taking expectation on both sides and applying Gronwall's inequality yields the result.
\end{customproof}

\section{Conclusion and Future Works}

In this work, we have interpreted the continuous-time MP dynamics in monotone games as a closed-loop equilibrium path in some specific (stochastic) differential games, referred to as mirror differential games (MDG). The construction of MDG is built on the Breizs-Ekeland variational principle, where the cost functional is determined by the mirror map in MP and the payoff in the monotone game. We provided finite-time characterization for this path, which indicates the equilibrium convergence in the time-average sense, and exponential convergence when the game is strongly monotone.

An immediate future research direction is the non-equilibrium path characterization in the stochastic case, i.e., to study the transient probabilistic behavior of MP dynamics. 
It is also of great interest to look at the effect of varying learning rates and mirror maps, which would pave the way for the design of multi-agent learning dynamics.

\bibliographystyle{ieeetr}
\bibliography{ref}

\end{document}